\documentclass[conference,letterpaper]{IEEEtran}

\usepackage[utf8]{inputenc} 
\usepackage[T1]{fontenc}
\usepackage{url}
\usepackage{ifthen}
\usepackage{cite}
\usepackage[cmex10]{amsmath} % Use the [cmex10] option to ensure complicance
\usepackage{dor_macros}
\newcommand{\infomat}{\rI^{X,Y}}
\newcommand{\infohat}{\hat{\rI}^{X,Y}}
\newcommand{\infohatGij}{\infohat_{G,i,j}}
\newcommand{\TE}{T^{X\to Y}}

\begin{document}
\title{InfoMat: A Tool for the Analysis and Visualization Sequential Information Transfer} 

% %%% Single author, or several authors with same affiliation:
% \author{%
%  \IEEEauthorblockN{Andrew R.~Barron}
%  \IEEEauthorblockA{Department of Statistics and Data Science\\
%                    Yale University\\
%                    New Haven, CT, USA\\
%                    Email: andrew.barron@yale.edu}
% }

%%% Several authors with up to three affiliations:
\author{%
  \IEEEauthorblockN{Dor Tsur}
  \IEEEauthorblockA{School of Electrical and Computer engineering\\
                    Ben-Gurion University of the Negev\\
                    Email: dortz@post.bgu.ac.il}
  \and
  \IEEEauthorblockN{Haim Permuter}
  \IEEEauthorblockA{School of Electrical and Computer engineering\\
                    Ben-Gurion University of the Negev\\
                    Email: haimp@bgu.ac.il}
}

%%% Many authors with many affiliations:
% \author{%
%   \IEEEauthorblockN{Andrew R.~Barron\IEEEauthorrefmark{1},
%                     Claude E.~Shannon\IEEEauthorrefmark{2},
%                     David Slepian\IEEEauthorrefmark{2},
%                     and Jacob Ziv\IEEEauthorrefmark{2}\IEEEauthorrefmark{3}}
%   \IEEEauthorblockA{\IEEEauthorrefmark{1}%
%                    Department of Statistics and Data Science, Yale University, New Haven, CT, USA,
%                     andrew.barron@yale.edu}
%   \IEEEauthorblockA{\IEEEauthorrefmark{2}%
%                     Bell Telephone Laboratories, Inc.,
%                     Murray Hill, NJ, USA,
%                     \{csh,dsl,jz\}@bell-labs.com}
%   \IEEEauthorblockA{\IEEEauthorrefmark{3}%
%                     Department of Electrical Engineering, Technion---Institute of Technology, Haifa, Israel,
%                     jz@ee.technion.ac.il}
% }

\maketitle

%%%%%%
%% Abstract: 
%% If your paper is eligible for the student paper award, please add
%% the comment "THIS PAPER IS ELIGIBLE FOR THE STUDENT PAPER
%% AWARD." as a first line in the abstract. 
%% For the final version of the accepted paper, please do not forget
%% to remove this comment!
%%

\begin{abstract}
   Despite the popularity of information measures in analysis of probabilistic systems, proper tools for their visualization are not common.
   This work develops a simple matrix representation of information transfer in sequential systems, termed information matrix (InfoMat).
   The simplicity of the InfoMat provides a new visual perspective on existing decomposition formulas of mutual information, and enables us to prove new relations between sequential information theoretic measures.
   We study various estimation schemes of the InfoMat, facilitating the visualization of information transfer in sequential datasets.
   By drawing a connection between visual pattern in the InfoMat and various dependence structures, we observe how information transfer evolves in the dataset.
   We then leverage this tool to visualize the effect of capacity-achieving coding schemes on the underlying exchange of information.
   We believe the InfoMat is applicable to any time-series task for a better understanding of the data at hand.
\end{abstract}

\section{Introduction}
Information theory is central to the analysis of stochastic systems.
Through the lens of information theory, one can study the evolution of temporal dependence in a sequential system, which is interpreted as the exchange of information.
For example, the interaction in a communication channel with feedback, that consists of an interacting encoder-channel pair, can be analysed through the directed information terms \cite{massey1990causality}, $I(X^m\to Y^m)$ and $I(Y^m\to X^m)$, which are given by
\begin{equation}\label{eq:di}
    I(X^m\to Y^m):=\sum_{i=1}^m I(X^i;Y_i|Y^{i-1}).
\end{equation}
These correspond to the information transfer over $m$ step from the encoder $(X^m)$ to the channel $(Y^m)$, and vice versa.
Thus, feedback capacity, defined as the maximization of normalized directed information \cite{tatikonda2008capacity}, represents optimization of information transfer from $X^m$ to $Y^m$.
Another pertinent example is neuroscience \cite{wibral2014directed}, in which information theory is widely used for the analysis of neurological data.
For neuroscience, the central measure of interest is transfer entropy \cite{schreiber2000measuring}.
For parameters $(m,k,l)$ the transfer entropy is given by
\begin{equation}\label{eq:te}
    T^{X\to Y}_m(k,l):= I(X^{m-1}_{m-k};Y_m|Y^{m-1}_{m-l}),
\end{equation}
and is interpreted as the effect of $k$ past $X$ symbols on the present symbol of $Y$, given $l$ elements of its past, calculated at time $m$.
Beyond communications and neuroscience, information theory was shown useful for various fields of sequential analysis, encompassing control \cite{klyubin2005empowerment,tanaka2017lqg,tiomkin2017unified,sabag2022reducing} ,causal inference \cite{raginsky2011directed,wieczorek2019information} and various machine learning tasks \cite{zhou2016causal,kalajdzievski2022transfer,bonetti2023causal}.

In contrast to its wide use to characterize and analyse information theoretic relations in time series, information theory fails to offer simple visualization tools to demonstrate its merits, and existing visualization techniques are not common. The cross-correlation matrix measures correlations between the components of two random vectors. However, its accuracy is limited to univariate Gaussian sequences with linear relations, and does not naturally generalize to conditional settings. The Venn Diagram \cite{CovThom06} also serves as a visualization of mutual information (MI), through decomposition into joint and conditional entropies. Furthermore the Information Diagram \cite{correa2013mutual}, which is based on the Taylor Diagram \cite{taylor2001summarizing}, explores the interplay between entropies and MI in a geometrical fashion. Beyond the visualization of information measures, information theory is widely used to evaluate other visualization and rendering  techniques \cite{chen2016information}.

We propose the information matrix (InfoMat), a novel visualization tool of the information transfer between two interacting sequences. 
The InfoMat arranges the involved conditional MI terms in an $m\times m$ matrix.
As a result, it provides a visual representation of existing information theoretic conservation laws and decompositions, while revealing new relationships.
The relationships are proved by characterizing different subsets of the matrix with corresponding information measures.
The InfoMat also serves as a visualization tool for arbitrary sequential data.
Using a heatmap representation, we can link various dependence structures in the data with visual patterns, and visualize the power of optimal coding schemes in channels with memory.
We propose a Gaussian approximation for the InfoMat that relies on the calculation of sample covariance matrices and analyze its theoretical guarantees, while empirically demonstrate the power of the resulting visualization tool.
When a Gaussian approximation is insufficient, we develop a neural InfoMat estimator, which is based on masked normalizing flows (MAFs), which expands the class of distributions captured by the InfoMat.

\section{The InfoMat}\label{sec:infomat}
We consider a pair of jointly distributed sequences $(X^m,Y^m)\sim P_{X^m,Y^m}$ such that $X_i=Y_i=\emptyset$ for $k\notin[1,\dots m]$. 
The InfoMat is an $m\times m$ matrix whose entries are given by conditional MI (CMI) terms, i.e., 
\begin{equation}
    \infomat\in\RR^{m\times m}_{\geq 0},\quad \infomat_{i,j} := I(X_i;Y_j|X^{i-1},Y^{j-1})
\end{equation}
The InfoMat captures the relations between the sequences elements, and its definition is motivated by the chain rule for MI \cite{CovThom06}, given by
$$
I(X^m;Y^m)=\sum_{i=1}^m\sum_{j=1}^m I(X_i;Y_j|X^{i-1},Y^{j-1})=\sum_{i,j}\infomat_{i,j},
$$
i.e., the MI between $X^m$ and $Y^m$ is given by the sum over all of the entries of $\infomat$.
We therefore interpret various decompositions of $I(X^m;Y^m)$ (termed information conservation laws \cite{massey2005conservation}) as grouping the entries of $\infomat$ into meaningful subsets, which sum to $I(X^m;Y^m)$.

\subsection{Capturing Information Conservation Laws}
Two particularly popular conservation laws were given by \cite{massey2005conservation} and \cite{amblard2011directed}.
These rules characterize the information transfer within a sequential setting, and are usually derived via algebraic manipulation of information measures. In some cases, these decompositions may be difficult to comprehend and visualize.
Using $\infomat$, we can visualize sequential measures. For example, directed information is given by.
\begin{equation}\label{eq:di_infomat}
\hspace{-0.18cm}I(X^m\to Y^m)=\mathbf{1}^\sT
    \begin{pmatrix}
{\color{BlueViolet}\infomat_{1,1}} & {\color{BlueViolet}\infomat_{1,2}} & {\color{BlueViolet}\dots} & {\color{BlueViolet}\infomat_{1,m}}\\
0 & {\color{BlueViolet}\infomat_{2,2}} & {\color{BlueViolet}\ddots}& {\color{BlueViolet}\vdots}\\
\vdots & \ddots &{\color{BlueViolet}\ddots} & {\color{BlueViolet}\infomat_{m-1,m}}\\
0& \dots & 0 & {\color{BlueViolet}\infomat_{m,m}}
\end{pmatrix} \mathbf{1},
\end{equation}
where $\mathbf{1}:=[1,\dots,1]\in\RR^m$.
Massey's conservation law decomposes MI into two 'opposite' directed information terms, with the purpose of dividing the information flow to and from the channel. It is given by \cite[Proposition~2]{massey2005conservation}
\begin{equation}\label{eq:conservation_massey_infomat}
    I(X^m;Y^m) = {\color{BlueViolet}I(X^{m}\to Y^m)} + {\color{Cerulean}I(\sD^1\circ Y^{m}\to X^m)}
\end{equation}
where $I(\sD^k\circ X^m\to Y^m):=\sum_{i=1}^m I(X^{i-k};Y_i|Y^{i-1})$ is the time-delayed directed information with $\sD^k\circ X^m$ being a left concatenation of $k$ null symbols with $X_{k+1}^{m}$.
Note that $I(X^{m}\to Y^m)$ corresponds to a triangular submatrix of $\infomat$, and the $\sD^k$ operation yields a triangular with side length $(m-k)$.
The InfoMat proposes a simple, alternative proof of \eqref{eq:conservation_massey_infomat} by coloring index subsets and summing over each color group.
\begin{equation}
\hspace{-0.18cm}I(X^m;Y^m)=\mathbf{1}^\sT
    \begin{pmatrix}
{\color{BlueViolet}\infomat_{1,1}} & {\color{BlueViolet}\infomat_{1,2}} & {\color{BlueViolet}\dots} & {\color{BlueViolet}\infomat_{1,m}}\\
{\color{Cerulean}\infomat_{2,1}} & {\color{BlueViolet}\infomat_{2,2}} & {\color{BlueViolet}\ddots}& {\color{BlueViolet}\vdots}\\
{\color{Cerulean}\vdots} &{\color{Cerulean}\ddots} &{\color{BlueViolet}\ddots} & {\color{BlueViolet}\infomat_{m-1,m}}\\
{\color{Cerulean}\infomat_{m,1}} &{\color{Cerulean}\dots} & {\color{Cerulean}\infomat_{m,m-1}} & {\color{BlueViolet}\infomat_{m,m}}
\end{pmatrix} \mathbf{1},
\end{equation}
The authors of \cite{amblard2011directed} propose a modification of \eqref{eq:conservation_massey_infomat} that distinguishes between past and present effect, given by
\begin{align}\label{eq:conservation_amblard_infomat}
    &\hspace{-0.12cm}I(X^m;Y^m)=\nonumber \\
    &\hspace{-0.12cm}{\color{Mahogany}I(\sD^1\circ X^m\to Y^m)} + {\color{Cerulean}I(\sD^1\circ Y^m\to X^m)} + {\color{ForestGreen}I_{\mathsf{inst}}(X^m,Y^m)},
\end{align}
where ${\color{ForestGreen}I_{\mathsf{inst}}(X^n,Y^n):=\sum_{i=1}^n I(X_i;Y_i|X^{i-1},Y^{i-1})}$ is the instantaneous information, which measures the \textit{symmetric} dependence between $X^n$ and $Y^n$, given a shared history.
Noting that $I_{\mathsf{inst}}(X^n,Y^n)=\mathsf{Trace}(\infomat)$, we can visualize \eqref{eq:conservation_amblard_infomat} as
\begin{equation}
    \hspace{-0.12cm}I(X^m;Y^m) = \mathbf{1}^\sT
    \begin{pmatrix}
{\color{ForestGreen}\infomat_{1,1}} & {\color{Mahogany}\infomat_{1,2}} & {\color{Mahogany}\dots} & {\color{Mahogany}\infomat_{1,m}}\\
{\color{Cerulean}\infomat_{2,1}} & {\color{ForestGreen}\infomat_{2,2}} & {\color{Mahogany}\ddots}& {\color{Mahogany}\vdots}\\
{\color{Cerulean}\vdots} &{\color{Cerulean}\ddots} &{\color{ForestGreen}\ddots} & {\color{Mahogany}\infomat_{m-1,m}}\\
{\color{Cerulean}\infomat_{m,1}} &{\color{Cerulean}\dots} & {\color{Cerulean}\infomat_{m,m-1}} & {\color{ForestGreen}\infomat_{m,m}}
\end{pmatrix}
\mathbf{1}
\end{equation}

\subsection{Developing New Information Theoretic Relations}
The simplicity of the InfoMat visualization allow us to develop new meaningful information theoretic equivalences.
The proofs draw upon the connections between information theoretic quantities and their corresponding indices subsets in $\infomat$ (e.g., \eqref{eq:di_infomat}).
We begin with the following proposition

\begin{proposition}[Transfer entropic decomposition of directed information]\label{prop:di_via_te}
    For $(X^m,Y^m)\sim P_{X^m,Y^m}$ and $1 \leq k \leq m$, we have
    $$
    I(\sD^k \circ X^m\to Y^m)=\sum_{i=1}^{m-k} \TE_{i+1}(i,i).
    $$
\end{proposition}
\begin{proof}
We provide the proof for $k=1$.
    Recall that directed information corresponds to the upper triangular part of $\infomat$.
    We note that $\TE_{i+1}(i,i)$ corresponds to a column in $\infomat$ that begins in the $(i+1)$th column and has length $i$.
Thus, we color the triangular (directed information) as follows
$$
\begin{pmatrix}
{\color{black}\hspace{0.5cm}} & {\color{blue}\infomat_{1,2}} & {\color{brown}\infomat_{1,3}}  & {\color{BlueViolet}\dots} & {\color{orange}\infomat_{1,m}}\\
{\color{Cerulean}} & {\color{ForestGreen}} & {\color{brown}\infomat_{2,3}}  & {\color{BlueViolet}\ddots}& {\color{orange}\vdots}\\
{\color{Cerulean}} & {\color{Cerulean}} & {\color{brown}}  & {\color{BlueViolet}\ddots}& {\color{orange}\vdots}\\
% {\color{Cerulean}\infomat_{4,1}} & {\color{Cerulean}\infomat_{4,2}} & {\color{Cerulean}\infomat_{4,3}}  & {\color{BlueViolet}\ddots}& {\color{BlueViolet}\vdots}\\
{\color{Cerulean}} &{\color{Cerulean}} &{\color{Cerulean}} &{\color{BlueViolet}} & {\color{orange}\infomat_{m-1,m}}\\
{\color{Cerulean}} &{\color{Cerulean}} &{\color{Cerulean}} & {\color{Cerulean}} & {\color{orange}}
\end{pmatrix}
$$
The relation then follows by summing over the upper triangular and dividing the sum into the corresponding rows, i.e.
\begin{align*}
    &I(\sD^1\circ X^m \to Y^m) \\
    &\hspace{-0.15cm}= {\color{blue}T_{2}^{X\to Y}(1,1)} + {\color{brown}T_{3}^{X\to Y}(2,2)} +\dots + {\color{orange}T_{m}^{X\to Y}(m-1,m-1)} 
\end{align*}
The proof similarly extends to $k>1$ with similar steps.
\end{proof}
\noindent In a similar fashion, we draw the following connections
% Having drawn a connection between sequential information theoretic quantities and their patterns in $\infomat$ \dt{(which are summarized in Table \ref{tab:infomat_mapping})}, we can reveal additional relations through the use of the InfoMat.
% We propose the following 
\begin{proposition}[Conservation of transfer entropy]\label{prop:te_conservation}
    \begin{equation*}\label{eq:te_conservation}
    I(X^m;Y^m) = \sum_{i=1}^{m-1} T^{X\to Y}_{i+1}(i,i) + T^{Y\to X}_{i+1}(i,i) + I_{\mathsf{inst}}(X^m,Y^m).
\end{equation*} 
\end{proposition}
\begin{proposition}[Directed information chain rule]\label{prop:di_chain_rule}
    \begin{align*}
    &I(\sD^k\circ X^m\to Y^m)=\\
    &\hspace{1.5cm}I(\sD^{k+1}\circ X^m\to Y^m) + \sum_{i=1}^m I(X_{i-k};Y_{i}|Y^{i-1}).
    \end{align*}
\end{proposition}
\noindent The proof of Propositions \ref{prop:te_conservation} and \ref{prop:di_chain_rule} follow by arguments and tools similar to Proposition \ref{prop:di_via_te}.
% applying Proposition \ref{prop:di_via_te} to \eqref{eq:conservation_amblard_infomat}, while Proposition \ref{prop:di_chain_rule} directly follows from Proposition \ref{prop:di_via_te} with similar coloring arguments.
The visual patters and corresponding information theoretic relations are summarized in Table \ref{tab:infomat_mapping}.

% Next, we propose a simple approximation for the InfoMat that is based on differential entropies of multivaraite Gaussian distributions.

\begin{table}[t]
    \centering
    \setlength{\extrarowheight}{1pt}
    \begin{tabular}{||c|c||}
    \hline
    Information measure & Visual pattern in $\infomat$\\
    \hline
       $I(X^m \to Y^m)$  & Upper triangular \textbf{with diagonal} \\
       \hline
        $I(\sD^k\circ X^m \to Y^m)$ & Upper triangular, side $(m-k)$ \\
        \hline
        $\TE_{i+1}(i,i)$ & Col. in row $i+1$ with length $i$ \\
        \hline
    \end{tabular}
    \caption{Visual shapes of dependence patterns in $\infomat$}
    \label{tab:infomat_mapping}
    \vspace{-0.5cm}
\end{table}

\section{Gaussian Approximation}\label{sec:gaussian_approx}
In this section we propose an approximation of $\infomat$ from samples when $P_{X^m,Y^m}$ is unknown, or its entries are not given in a closed form.
As we are required to estimate $m^2$ CMI terms, we propose a Gaussian approximation of CMI and study its performance. In a later section, we propose a neural estimator of $\infomat$ for a case where a Gaussian approximation is insufficient.

\subsection{Proposed Estimator}
Let $(X^n,Y^n)\sim P_{X^n,Y^n}$ be a given dataset from which we want to estimate $\infomat$.
For simplicity of presentation, we assume all variables have zero mean.
We begin by using the following representation of CMI
\begin{lemma}\label{lemma:gaussian_form}
Let $(X^m,Y^m)\sim P_{X^m,Y^m}$, and $i,j
\leq m$. Then,
    \begin{align}
        I(X_i;Y_j|X^{i-1},Y^{j-1})&=H(X^{i},Y^{j-1})+H(X^{i-1},Y^{j})\nonumber \\
        &\hspace{-0.025cm}- H(X^{i-1},Y^{j-1}) - H(X^{i},Y^{j}).\label{eq:cmi_entropies}  
    \end{align}
    If $(X^m,Y^m)$ are jointly Gaussian, then
    \begin{equation}\label{eq:cmi_entropies_gaussian}
        \hspace{-0.15cm}I(X_i;Y_j|X^{i-1},Y^{j-1})=\frac{1}{2}\log\frac{\left|\rK_{X^{i},Y^{j-1}}\right|\left|\rK_{X^{i-1},Y^{j}}\right|}{\left|\rK_{X^{i-1},Y^{j-1}}\right|\left|\rK_{X^{i},Y^{j}}\right|},
    \end{equation}
    where $\rK_Z$ is the covariance matrix of $Z\sim P_Z$, and $|\rK_Z|$ is its determinant. 
\end{lemma}
In what follows, we denote the right-hand side of \eqref{eq:cmi_entropies_gaussian} with $\infomat_{G,i,j}$, which is a Gaussian approximation of $\infomat_{i,j}$.

Estimating $\infomat_{G,i,j}$ is significantly simpler than estimating $\infomat_{i,j}$, as the estimation problem boils down to the calculation of sample covariance matrices, which is a well studied problem.
An estimator of $\infomat_{i,j}$ from $(x^n,y^n)$, denoted $\infohatGij(x^n,y^n)$ is given by calculation of the Gaussian approximation \eqref{eq:cmi_entropies_gaussian} using sample covariance matrices, e.g. $\hat{\rK}_{i,j}:=\hat{\rK}_{X^i,Y^j}$.
To estimation $\hat{\rK}_{i,j}$ we divide $(x^n,y^n)$ into a dataset $(x^i,y^j)_{l=1}^N$ with $n/m \leq N \leq n$.
The Gaussian approximation procedure is summarized in Algorithm \ref{alg:infomat_est}.
Finally, a Gaussian estimator of $\infomat$ is an $m\times m$ matrix whose $(i,j)$ entry is given by $\infohatGij$.
%%%%
%%%%
\begin{algorithm}[ht]
\caption{Gaussian InfoMat Estimation}
\label{alg:infomat_est}
\textbf{input:} Data $(x^n,y^n)$, matrix length $m$\\
\textbf{output:} 
Gaussian estimate of $\infomat$
\algrule
\begin{algorithmic}
\State Initialize $\infohatGij=0$ for $(i,j)\in(1,\dots,m)\times(1,\dots,m)$
\For{$(i,j)$ in $(1,\dots,m)\times(1,\dots,m)$}
    \State Divide $(X^n,Y^n$) into datasets 
    $$
    ((x^{i-1},y^{j-1})_l,(x^{i},y^{j-1})_l,(x^{i-1},y^{j})_l,(x^{i},y^{j})_l)_{l=1}^N
    $$
    \State Calculate sample covariance matrices 
    \State Calculate $\infohatGij$ via \eqref{eq:cmi_entropies_gaussian}.
\EndFor
\Return Estimated InfoMat.
\end{algorithmic}
\end{algorithm}
%%%%
%%%%
% To summarize, estimating $\infohatGij$ from $(x^n,y^n)$ involves
% \begin{enumerate}
%     \item Obtaining datasets 
%     $$((x^{i-1},y^{j-1})_l,(x^{i},y^{j-1})_l,(x^{i-1},y^{j})_l,(x^{i},y^{j})_l)_{l=1}^N
%     $$
%     \item Plugging into Gaussian approximation formula \eqref{eq:cmi_entropies_gaussian} \dt{consider writing in alg form instead} 
% \end{enumerate}

The proposed Gaussian approximation for $\infomat$ relies on the estimation of covariance matrices.
Therefore, it inherits its guarantees from those of log determinants of sample covariance matrices. We thus have the following
\begin{proposition}[Gaussian approximation performance guarantees]\label{prop:gaussian_estiamtor}
    Let $(X^{N\times m},Y^{N\times m})$ be a set of $N$ length $m$ sequences of jointly $(d_x,d_y)$-dimensional Gaussian random vectors.
    Then
    \begin{enumerate}
        \item Bias: $\quad\lim_{n\to\infty}\EE\left[\infohatGij\right] = \rI_{G,i,j}.$
        \item Variance: $\quad\lim_{n\to\infty}\mathsf{Var}(\infohatGij) = O(\frac{dm}{N}) = O(\frac{1}{N}).$
    \end{enumerate}
    where $d=\max(d_x,d_y)$.
\end{proposition}
\noindent The proof follows from the analysis in \cite{cai2015law} and \cite{duong2023diffeomorphic}.

\section{Beyond Gaussian - Neural Estimation}
When the Gaussian approximation is insufficient (due to e.g. high non-linearities or non-Gaussianity), we propose a neural estimation scheme for $\infomat$, which relies on a novel CMI estimation procedure.
Due to space constraints, we provide a high-level description of the estimator and its application to the InfoMat, reserving full details for the extended version.

\begin{figure}[b!]
    \centering
    \psfrag{A}[][][0.9]{\hspace{0.15cm}$(x^i,y^j)$}
    \psfrag{B}[][][1.2]{MAF}
    \psfrag{C}[][][1]{\vspace{0.2cm}$f_\theta^{-1}$}
    \psfrag{D}[][][0.9]{\hspace{0.1cm}$(x'_i,y'_j)$}
    \psfrag{E}[][][1.2]{$\Phi^{-1}$}
    \psfrag{F}[][][1.2]{$\hat{\rK}_{i,j}$}
    \psfrag{G}[][][1.2]{GMI}
    \psfrag{H}[][][1]{\hspace{0.15cm}$\hat{\infomat}_{i,j}$}
    \psfrag{I}[][][1]{ML loss, $\cL_{\theta}$}
    \psfrag{J}[][][1]{\hspace{0.2cm}$\nabla_\theta \cL_\theta$}
    \psfrag{K}[][][1]{}
    \includegraphics[trim={0pt 0pt 0pt 0pt}, clip,width=0.97\linewidth]{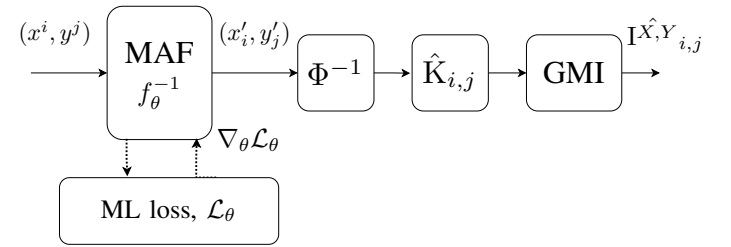}
    \caption{Neural estimation model. Dashed line represents maximum-likelihood (ML) training phase, while the filled lines account for the inference (MI calculation) phase.}
    \label{fig:ne_model}
\end{figure}

% % \subsection{Diffeomorphic Mutual Information Estimator}
We employ the Diffeomorphic MI (DMI) estimator \cite{duong2023diffeomorphic}, which leverages the power of conditional MAFs \cite{papamakarios2021normalizing} for the task of CMI estimation.
MAFs fit a distribution $p_\theta(X)$ by mapping samples of a base distribution $p(U)$ with a parametric diffeomorphism $f_\theta$ (i.e., a differentiable, invertible function with a differentiable inverse).
MAFs consider functions $f_\theta$ with triangular Jacobian, which provide simple change of variable formula for $p_\theta(x)$.
Such mappings can be realized by neural networks with masked weight matrices and monotone activations \cite{papamakarios2017masked}.
When $p(U)=\mathsf{Unif}[0,1]^d$, we have
\begin{equation}\label{eq:ml_loss}
    \log p_\theta(x) = \log(d) + \sum_{i=1}^d\log\left| \frac{\partial u_i}{\partial x_i} \right|,
\end{equation}
which serves as a maximum-likelihood loss to optimize $\theta$.
This scheme adapts to conditional distributions $p_\theta(x|z)$ using conditioner models $g_{\theta'}(z)$, such that the compound model is given by
$x_i = f_\theta(u) + g_{\theta'}(z)$, and is termed a conditional diffeomorphism.
For more information on MAFs, we refer the reader to \cite{kobyzev2020normalizing}.

The DMI uses conditional MAFs to estimate $I(X;Y|Z)$ mapping them into Gaussian variables $(X',Y')$, which enables calculation of CMI via a Gaussian approximation.
In the DMI training, we optimize $f_\theta^{-1}$, which learns a mapping of $X$ and $Y$ into a uniform variable $U$, which is then mapped to a Gaussian variables via the inverse Gaussian cumulative distribution function $\Phi^{-1}$.
% The CMI of the resulting jointly Gaussian pair can be easily approximated with sample covariance matrices.
Consequently, the Gaussian CMI is claimed to be equal due to the invariance of CMI to diffeomorphisms:
\begin{proposition}[Conditional mutual information invariance]\label{prop:cmi_inv}
    Let $(X,Y,Z)\sim P_{X,Y,Z}$ and denote by $(X',Y')=(f_\theta(X,Z),g_\phi(Y,Z))$, where $f_\theta$ and $g_\phi$ are conditional diffeomorphisms. Then,
    \begin{equation}
        I(X;Y|Z) = I(X';Y'|Z)
    \end{equation}
\end{proposition}
Proposition \ref{prop:cmi_inv} is a slight modification of \cite[Lemma~2]{duong2023diffeomorphic}.
The existence of optimal MAFs is guaranteed by the universal approximation properties of normalizing flows \cite{papamakarios2021normalizing}.
The MAF-based scheme is depicted in Figure \ref{fig:ne_model}.

To estimate $\infomat$, we train a DMI model for the estimation of each entry.
The data is split in a similar fashion to Algorithm \ref{alg:infomat_est}, but due to the parametric nature of the estimator, we optimize using iterative minibatch-gradient descent.
In the training phase we optimize all DMI models in parallel through the optimization of the maximum-likelihood loss \ref{eq:ml_loss} for a fixed number of epochs.
When the training concludes, for each $\infomat_{i,j}$ we feed the corresponding dataset through the optimized MAFs, estimate sample covariance matrices, from which we calculate the MI term.
Neural network optimization is considerably expensive.
However, \cite{duong2023diffeomorphic} show that the DMI outperforms existing CMI estimators in terms of sample requirements.

%%%%%%%%%%%%%%%%%%%%
%%%%%%%%%%%%%%%%%%%%
%%%%%%%%%%%%%%%%%%%%
\begin{figure*}[t!]
    \centering
    \begin{subfigure}[ht]{0.32\linewidth}
        \includegraphics[width=\linewidth]{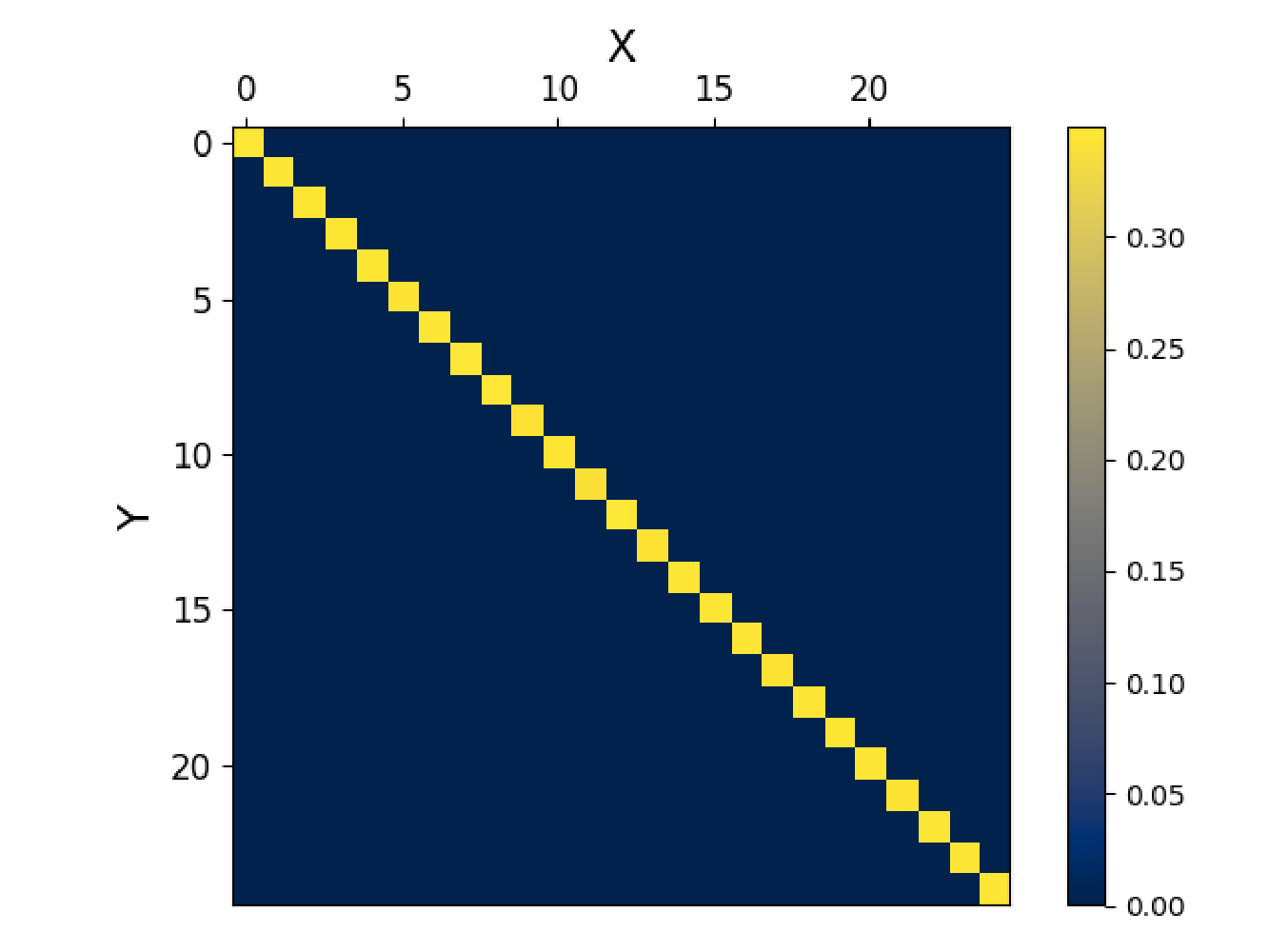}
        \caption{Gaussian i.i.d. data.}
        \label{fig:infomat_g_iid}
    \end{subfigure}
    \hfill
    \begin{subfigure}[ht]{0.32\linewidth}
        \includegraphics[width=\linewidth]{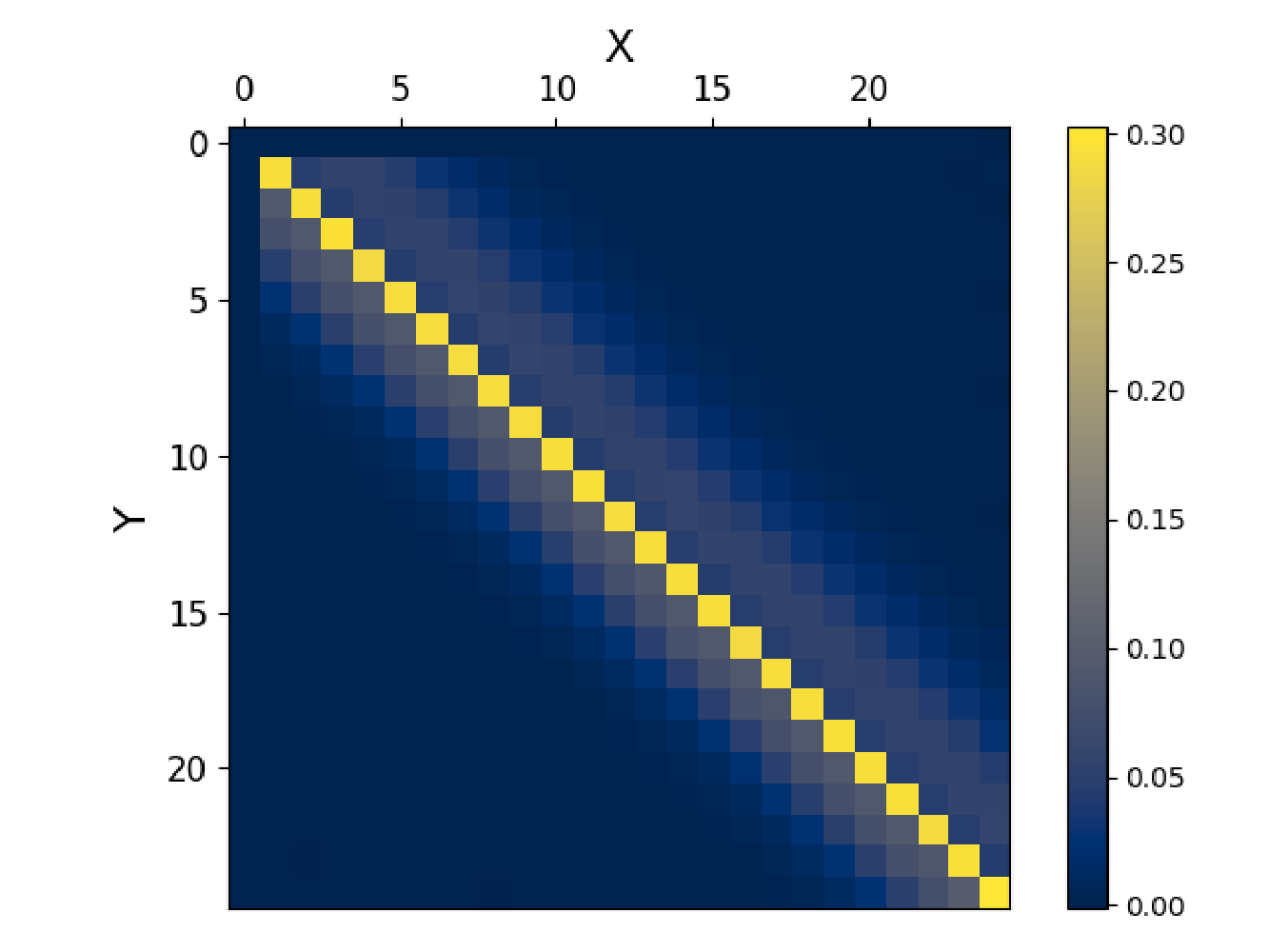}
        \caption{Gaussian ARMA, $\gamma=0.5$}
        \label{fig:infomat_g_arma}
    \end{subfigure}
    \hfill
    \begin{subfigure}[ht]{0.32\linewidth}
        \includegraphics[width=\linewidth]{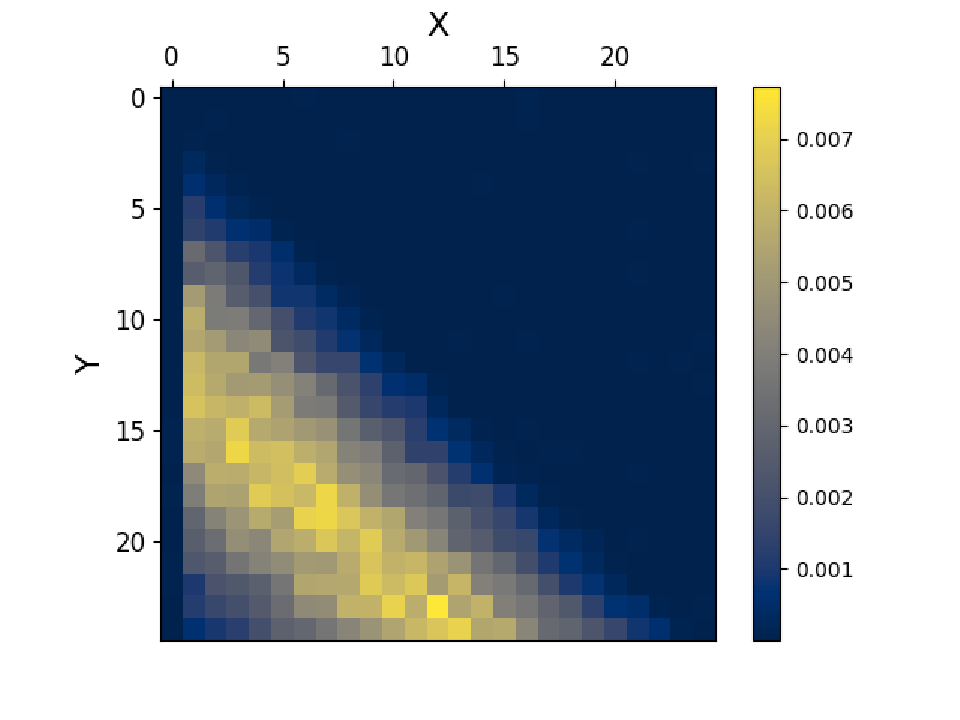}
        \caption{Gaussian ARMA, increasing weights.}
        \label{fig:infomat_g_delayed_arma}
    \end{subfigure}
    \caption{Visual patterns in the InfoMat for several ARMA processes over continuous spaces.}
\end{figure*}
%%%%%%%%%%%%%%%%%%%%
%%%%%%%%%%%%%%%%%%%%
%%%%%%%%%%%%%%%%%%%%

\section{Visualization of Information Transfer}\label{sec:visualization}
In this section we demonstrate the utility of the InfoMat as a visualization tool for sequential data.
By leveraging the connections from Section \ref{sec:infomat} between various parts of $\infomat$ and meaningful information transfer measures, we investigate the evolution of information transfer in given datasets.
We adopt the interpretation that higher directed information in a certain direction implies higher causal effect \cite{jiao2013universal,tsur2023neural}.
However, we provide a finer granularity of observation by looking at the CMI of every time-step.
Leveraging the developed InfoMat estimators, we propose visualizations of $\infomat$ as a heatmap, validating the connection between dependence structures and visual patterns.
An implementation of all settings can be found at \href{https://github.com/DorTsur/infomat}{GitHub}.

\subsection{Continuous Data - Gaussian Approximation}
Consider a Gaussian autoregressive (AR) setting, given by
\begin{align*}
    X_t = \sum_{k=0}^{t-1} \alpha^X_k X_{t-k} + \alpha^Y_k Y_{i-k} + N^X_t\\
    Y_t = \sum_{k=0}^{t-1} \beta^X_k X_{t-k} + \beta^Y_k Y_{t-k} + N^Y_t,\numberthis{}{}\label{eq:gauss_arma}
\end{align*}
where $N^Y_t$ and $N^Y_t$ are samples a of centered i.i.d. Gaussian processes with covariance matrices $\rK_{n_x}$ and $\rK_{n_y}$, respectively.
By controlling the values the AR model weights $(\alpha^X_k,\alpha^Y_k,\beta^X_k,\beta^Y_k)_{k=1}^m$, we induce different dependence structure on $(X^m,Y^m)$, which we then visualize via $\infomat$.
All visualizations in this subsection are obtained via Algorithm \ref{alg:infomat_est} with $n\approx10^5$ samples.

We begin with a simple i.i.d, setting by taking $\beta^X_k=\gamma\in(0,1)$ and nullifying the rest of the parameters. In this case $X_i\indep Y_j$ for $i\neq j$. 
As seen in Figure \ref{fig:infomat_g_iid}, the corresponding InfoMat captures the dependence structure, as we result with a diagonal matrix.
Next, we consider a symmetric dependence structure with memory by taking $\alpha^X_j=\alpha^Y_j=\beta^X_j=\beta^Y_j=\lambda\in(0,0.5)$.
As seen in Figure \ref{fig:infomat_g_arma}, the resulting InfoMat has most of its density on the diagonal, with a decay along the anti-diagonal.
This pattern is expected, as moving along the anti-diagonal from the diagonal is equivalent to measuring dependence w.r.t. further history, which decays due to multiplicative nature of the setting.

Finally, we consider a setting in which the parameter values increase as we look deeper back. In this case, there's a tradeoff between the increase in parameter value and the system's multiplicative decay.
This results in a sleeve of density around the shifted diagonal in which the shift between the two occur.
Additionally, we choose the parameter values such that the information transfer is in the direction $Y\to X$.
The corresponding InfoMat is given in Figure \ref{fig:infomat_g_delayed_arma}, which presents the 'information sleeve' due to the aforementioned tradeoff, and places it in the bottom triangular, as expected from the $Y\to X$ directional effect.
The diagonal on which the InfoMat attains a maxima corresponds to the time-step on which the shift occurs.

\subsection{Neural Estimator} To demonstrate the utility of neural estimation, and compare with the Gaussian approximation, we take a dataset with high nonlinearities.
Specifically, we take an i.i.d. jointly Gaussian dataset with correlation coefficient $\rho=0.9$.
We apply a cyclic shift of $T<m$ to the samples $Y^n$ within each $m$-length sequence, and employ a pair of nonlinear diffeomorphisms $X_i\mapsto \log(X_i)$ and $Y_i\mapsto Y_i^3$.
As visualized in Figure \ref{fig:nonlin_compare}, the neural estimator successfully captures the relations in the dataset, while the Gaussian approximation fails to provide a meaningful visualization.
As the mappings are invertible, the resulting MI is $0.83$ [nats] on nonzero entries, which are approximately the corresponding values in the neural estimator of $\infomat$.
However, this accuracy comes at the cost of training $m^2$ neural nets, which is significantly slower than calculating sample covariance matrices.

\begin{figure}[ht!]
    \hspace{-0.1cm}
    \begin{subfigure}[ht]{0.5\linewidth}
        \includegraphics[trim={30pt 30pt 30pt 1pt}, clip,width=\linewidth]{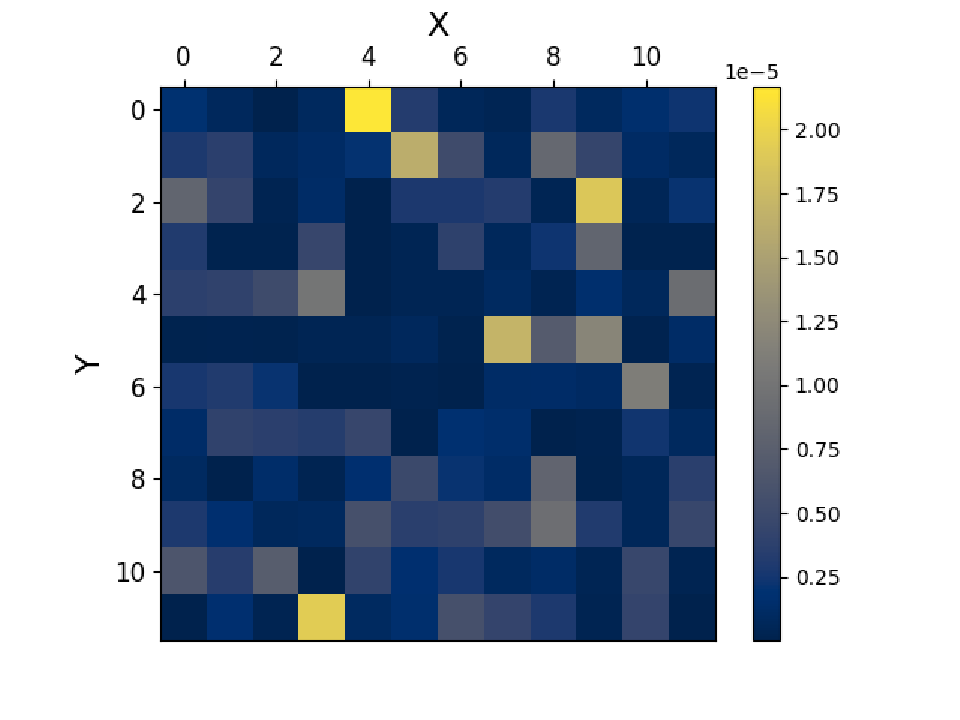}
        \caption{Gaussian approximation.}
        \label{fig:nonlin_gmi}
    \end{subfigure}
    \hspace{-0.3cm}
    \begin{subfigure}[ht]{0.5\linewidth}
        \includegraphics[trim={30pt 30pt 30pt 1pt}, clip,width=\linewidth]{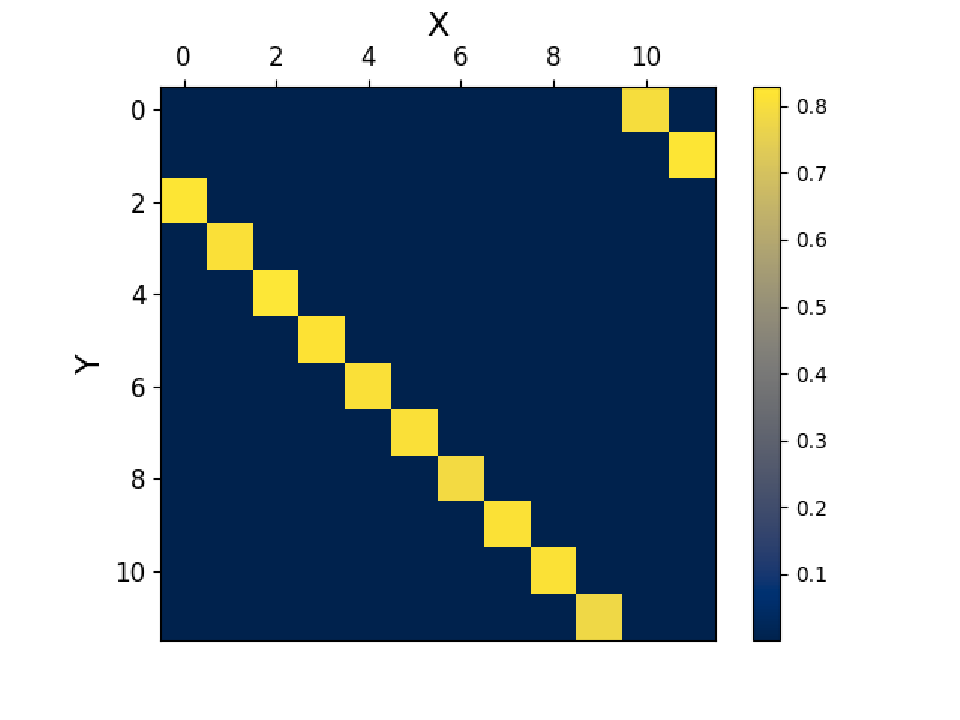}
        \caption{Neural estimator.}
        \label{fig:nonlin_ne}
    \end{subfigure}
    \caption{Estimated InfoMat under nonlinearities and cyclic shift.}
    \label{fig:nonlin_compare}
\end{figure}

\subsection{Visualizing Coding Schemes}
We now show the utility of the InfoMat for visualizing the effect of coding schemes in communication channels.
To present empirical results we consider a standard estimator of CMI, by applying a plug-in entropy estimator to the representation \eqref{eq:cmi_entropies}.
We consider the binary Ising channel with feedback, whose capacity has been analytically solved and a coding scheme was proposed in \cite{elishco2014capacity}.
As feedback capacity achieving distributions seek to maximize the normalized directed information from channel inputs to outputs \cite{tatikonda2008capacity}, we expect a feedback scheme to result with higher directed information.

We estimate the InfoMat in the Ising channel under two coding schemes.
The first is oblivious of any past channel inputs and outputs, sending $X^m\stackrel{i.i.d.}{\sim}\mathsf{Ber}(1/2)$ (Figure \ref{fig:ising_iid}). The second generates $X^m$ according to the feedback capacity achieving scheme \cite{elishco2014capacity} (Figure \ref{fig:ising_opt}). 
We note that the i.i.d. scheme generates an InfoMat with most of its information in the main diagonal and a small residue in the off-diagonal.
The diagonal information follows from the i.i.d. scheme, and it is constant along all time-steps, as it yields a time-invariant distribution. 
The off-diagonal information is due to the Ising state evolution $S_{i+1}=X_i$.
However, the optimal coding scheme generates a non-constant InfoMat.
Specifically, we observe that most of the information is sent along the off diagonal, i.e., most of the information is sent through the effect of past channel inputs.
Additionally, we note that the amount of conveyed information is non-constant. This is a result of the underlying finite-state machine that defines the evolution of $X^m$ according to past inputs, outputs and states \cite[Fig.~5]{elishco2014capacity}.
The estimated normalized directed information is
$$
\frac{1}{m}\hat{I}_{\mathsf{i.i.d.}}(X^m\to Y^m) \approx 0.45,\quad \frac{1}{m}\hat{I}_{\mathsf{opt}}(X^m\to Y^m) \approx 0.56
$$
Finally, we note that in contrast to the i.i.d. scheme, the optimal scheme generates information in the direction $Y\to X$ as well.

\begin{figure}[ht!]
    \hspace{-0.1cm}
    \begin{subfigure}[ht]{0.5\linewidth}
        \includegraphics[trim={30pt 30pt 30pt 1pt}, clip,width=\linewidth]{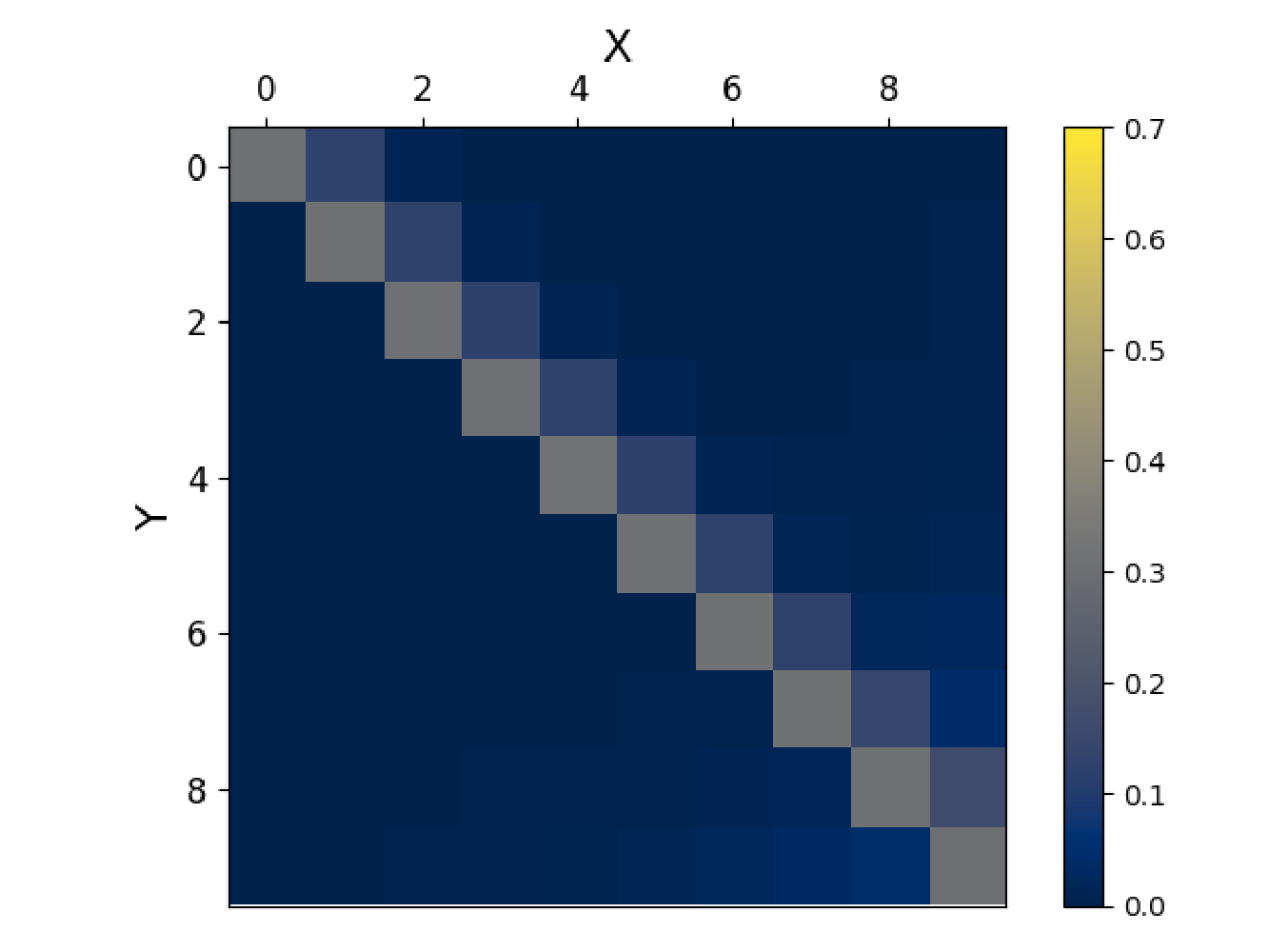}
        \caption{Oblivious coding scheme.}
        \label{fig:ising_iid}
    \end{subfigure}
    \hspace{-0.3cm}
    \begin{subfigure}[ht]{0.5\linewidth}
        \includegraphics[trim={30pt 30pt 30pt 1pt}, clip,width=\linewidth]{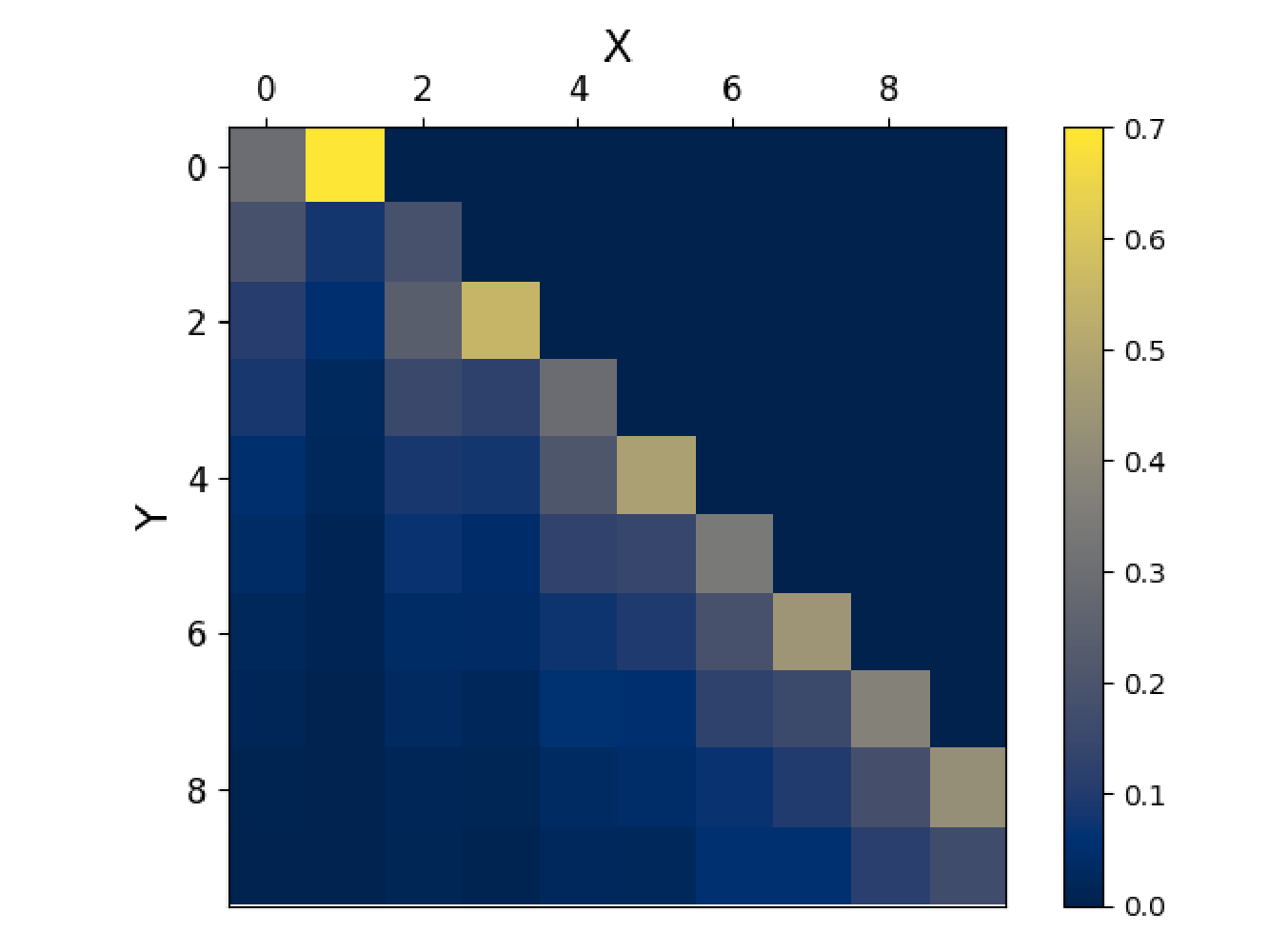}
        \caption{Optimal coding scheme.}
        \label{fig:ising_opt}
    \end{subfigure}
    \caption{Visualization of information transfer in the Ising channel under various coding schemes.}
    \label{fig:__}
\end{figure}

% \begin{figure}[ht]
%     \includegraphics[trim={0pt 0pt 0pt 0pt}, clip,scale=0.35]{Figs/ISIT/gauss_iid_isit.eps}
%     % \vspace{-0.1cm}
%     \caption{InfoMat Visualization.}
%     \label{fig:infomat_g_iid}
% \end{figure}

% \dt{We demonstrate the performance of the DMI estimator by considering a highly nonlinear sequential setting. We take the ARMA case from (before-ref) and apply a nonlinear mapping (fill). The resulting MI remains unchanged, as opposed to the simple linear estimator.}

% \dt{Add Alg.}

% \begin{figure}[t!]
%     \centering
%     \includegraphics[trim={30pt 20pt 1pt 1pt}, clip,width=0.7\linewidth]{Figs/ISIT/gauss_ne.eps}
%     \caption{Neural estimation of shifted data.}
%     \label{fig:gauss_ne}
% \end{figure}

% \dt{
% \begin{enumerate}
%     \item iid data
%     \item shifted iid
%     \item ARMRA - fat
%     \item ARMRA - condensed in one side
% \end{enumerate}
% }

\section{Conclusion}
This work developed the InfoMat, a matrix representation of information exchange.
We showed the InfoMat utility for visual proofs of information decomposition formulas by matrix coloring.
Then, we proposed several estimators of its entries and applied them for the visualization of information transfer in sequential systems.
For future work, we aim to develop a computationally efficient neural estimator of the InfoMat via weight sharing and slicing techniques \cite{tsur2024max}.
% Due to the simplicity of this work, and the popularity of information measures, the InfoMat can serve as a simple data exploration tool in sequential data analysis piplines.
Given this work's simplicity and the popularity of information measures, the InfoMat can serve as an effective tool for data exploration in sequential data analysis pipelines.
Furthermore, we believe that the InfoMat can be highly useful for a myriad of contemporary research fields that involve time-series.
Such fields encompass empowerment \cite{klyubin2005empowerment}, which characterizes robust sequential decision making via information theory, and causal inference \cite{raginsky2011directed}, in which information theory was shown beneficial.

\bibliographystyle{unsrt}
\bibliography{bibliography.bib}

\clearpage
\appendix

\begin{figure*}[!b]
\begin{equation}\label{eq:prop_di_chain_mat}
\begin{pmatrix}
{\color{BlueViolet}\infomat_{1,1}} & {\color{BlueViolet}\infomat_{1,2}} & {\color{BlueViolet}\infomat_{1,3}}  & {\color{BlueViolet}\dots} & {\color{BlueViolet}\infomat_{1,n}}\\
{\color{Cerulean}} & {\color{BlueViolet}\infomat_{2,2}} & {\color{BlueViolet}\infomat_{2,3}}  & {\color{BlueViolet}\ddots}& {\color{BlueViolet}\vdots}\\
{\color{Cerulean}} & {\color{Cerulean}} & {\color{BlueViolet}\infomat_{3,3}}  & {\color{BlueViolet}\ddots}& {\color{BlueViolet}\vdots}\\
% {\color{Cerulean}\infomat_{4,1}} & {\color{Cerulean}\infomat_{4,2}} & {\color{Cerulean}\infomat_{4,3}}  & {\color{BlueViolet}\ddots}& {\color{BlueViolet}\vdots}\\
{\color{Cerulean}} &{\color{Cerulean}} &{\color{Cerulean}} &{\color{BlueViolet}\ddots} & {\color{BlueViolet}\infomat_{n-1,n}}\\
{\color{Cerulean}} &{\color{Cerulean}} &{\color{Cerulean}} & {\color{Cerulean}} & {\color{BlueViolet}\infomat_{n,n}}
\end{pmatrix} 
=
\begin{pmatrix}
{\color{magenta}\infomat_{1,1}} & {\color{violet}\infomat_{1,2}} & {\color{violet}\infomat_{1,3}}  & {\color{violet}\dots} & {\color{violet}\infomat_{1,n}}\\
{\color{Cerulean}} & {\color{magenta}\infomat_{2,2}} & {\color{violet}\infomat_{2,3}}  & {\color{violet}\ddots}& {\color{violet}\vdots}\\
{\color{Cerulean}} & {\color{Cerulean}} & {\color{magenta}\infomat_{3,3}}  & {\color{violet}\ddots}& {\color{violet}\vdots}\\
% {\color{Cerulean}\infomat_{4,1}} & {\color{Cerulean}\infomat_{4,2}} & {\color{Cerulean}\infomat_{4,3}}  & {\color{BlueViolet}\ddots}& {\color{BlueViolet}\vdots}\\
{\color{Cerulean}} &{\color{Cerulean}} &{\color{Cerulean}} &{\color{magenta}\ddots} & {\color{violet}\infomat_{n-1,n}}\\
{\color{Cerulean}} &{\color{Cerulean}} &{\color{Cerulean}} & {\color{Cerulean}} & {\color{magenta}\infomat_{n,n}}
\end{pmatrix}  
\end{equation}
\end{figure*}

\subsection{Proof of Proposition \ref{prop:te_conservation}}
The proof utilizes the observation of Proposition \ref{prop:di_via_te}. We decompose MI, which is given by the entire matrix, into the upper and lower sub-triangulars (excluding the main diagonal), and the main diagonal.
As noted in the main text, $I_{\mathsf{inst}}(X^n,Y^n)$ corresponds to the main diagonal (black), $T_{i+1}^{X\to Y}(i,i)$ corresponds to a sub-column and $T_{i+1}^{Y\to X}(i,i)$ corresponds to a sub-row.
We therefore have
\begin{equation}
        \begin{pmatrix}
{\color{black}\infomat_{1,1}} & {\color{ForestGreen}\infomat_{1,2}} & {\color{Cerulean}\infomat_{1,3}} & {\color{orange}\infomat_{1,4}} & {\color{BlueViolet}\dots} & {\color{brown}\infomat_{1,n}}\\
{\color{red}\infomat_{2,1}} & {\color{black}\infomat_{2,2}} & {\color{Cerulean}\infomat_{2,3}} & {\color{orange}\infomat_{2,4}} & {\color{BlueViolet}\ddots}& {\color{brown}\vdots}\\
{\color{BlueViolet}\infomat_{3,1}} & {\color{BlueViolet}\infomat_{3,2}} & {\color{black}\infomat_{3,3}} & {\color{orange}\infomat_{3,4}} & {\color{BlueViolet}\ddots}& {\color{brown}\vdots}\\
{\color{violet}\infomat_{4,1}} & {\color{violet}\infomat_{4,2}} & {\color{violet}\infomat_{4,3}} & {\color{black}\infomat_{4,4}} & {\color{BlueViolet}\ddots}& {\color{brown}\vdots}\\
{\color{pink}\vdots} &{\color{pink}\ddots} &{\color{pink}\ddots} &{\color{pink}\ddots} &{\color{black}\ddots} & {\color{brown}\infomat_{n-1,n}}\\
{\color{magenta}\infomat_{n,1}} &{\color{magenta}\dots} &{\color{magenta}\dots} &{\color{magenta}\dots}  & {\color{magenta}\infomat_{n,n-1}} & {\color{black}\infomat_{n,n}}
\end{pmatrix}
\end{equation}
\begin{align*}
    &I(X^n;Y^n) \\
    &= ({\color{ForestGreen}T_{2}^{X\to Y}(1,1)} + {\color{Cerulean}T_{3}^{X\to Y}(2,2)} + \dots + {\color{brown}T_{n}^{X\to Y}(n-1,n-1)})\\
    &+ I_{\mathsf{inst}}(X^n,Y^n)\\
    &+ ({\color{red}T_{2}^{Y \to X}(1,1)} + {\color{BlueViolet}T_{3}^{Y\to X}(2,2)} + \dots + {\color{magenta}T_{n}^{Y\to X}(n-1,n-1)} )
\end{align*}

\subsection{Proof of Proposition \ref{prop:di_chain_rule}}
The relation follows from noting that a delayed directed information term $I(D^{k+1}\circ X^n\to Y^n)$ corresponds to a subtriangular element, which forms the one-step reduced directed information element $I(D^{k}\circ X^n\to Y^n)$ when combined with the appropriate sub-diagonal, which, in turn, correspond to a 'delayed' instantaneous MI term. For example, when $k=0$, it is given by \eqref{eq:prop_di_chain_mat},
which we by coloring and elements garthering decomposes as
$$
{\color{BlueViolet}I(X^n\to Y^n)} = {\color{magenta}I_{\mathsf{inst}}(X^n, Y^n)} + {\color{violet}I(\sD^1\circ X^n\to Y^n)}
$$

\subsection{Proof of Proposition \ref{prop:gaussian_estiamtor}}
Let $Z^n$ be a set of $n$ samples from $\cN(0,\rK_Z)$, defined on $\RR^d$.

\textbf{Bias:} According to \cite{cai2015law} the bias of $\log|\rK_Z|$ is given by  
    $$
    \tau_{n,d}:=\sum_{k=1}^d\left( \psi(\frac{n-k+1}{2})\log\frac{n}{2} \right),
    $$
where $\psi(z)$ is the digamma function, which is known to asymptotically behave as $\psi(z)\approx \log(z)-\frac{z}{2}$.
Thus, the bias $\tau_{n,d}$ term behaves as $1/n$, vanishing at $n\to\infty$.

\textbf{Variance:} We have the following:
\begin{align*}
    \mathsf{Var}(\infohatGij) &\leq \mathsf{Var}(\log|\hat{\rK}_{X^{i},Y^{j-1}}|) +  \mathsf{Var}(\log|\hat{\rK}_{X^{i-1},Y^{j}}|)\\
    &+  \mathsf{Var}(\log|\hat{\rK}_{X^{i-1},Y^{j-1}}|) +  \mathsf{Var}(\log|\hat{\rK}_{X^{i},Y^{j}}|).
\end{align*}
Thus, the estimator variance is governed by the variance of the log-determinant estimator.
Following the analysis in the proof of \cite[Lemma~6]{duong2023diffeomorphic}, 
$$
\mathsf{Var}(\log|\rK_Z|)\stackrel{L}{\longrightarrow} O\left(\frac{d_Z}{n}\right)
$$ 
In our case, the dimension of $\rK_{X^i,Y^j}$ is $d_xi+d_yj$, which is upper bounded by $2dm$, and the number of samples for the estimation of $\rK_{X^i,Y^j}$ is $n/\min(i,j)$, which is lower bounded by $n/m$.
As these bounds hold for all four sample covariance matrices, we have
$$
\mathsf{Var}\left(\infohatGij\right) = O\left(\frac{dm}{n/m}\right) = O\left(\frac{dm^2}{n}\right),
$$
which is sharp in $n$.

\subsection{Proof of Lemma \ref{lemma:gaussian_form}}
Let $(X^m,Y^m)\sim P_{X^m,Y^m}$.
    The entropy decomposition of conditional mutual information \eqref{eq:cmi_entropies} follows from the following steps
    \begin{align*}
        &I(X_i;Y_j|X^{i-1},Y^{j-1}) \\
        &= H(X_i|X^{i-1},Y^{j-1}) + H(Y_j|X^{i-1},Y^{j-1}) \\
        &\qquad- H(X_i,Y_j|X^{i-1},Y^{j-1})\\
        &=H(X^{i},Y^{j-1}) - H(X^{i-1},Y^{j-1}) + H(X^{i-1},Y^{j}) \\
        &\qquad- H(X^{i-1},Y^{j-1}) - (H(X^{i},Y^{j}) - H(X^{i-1},Y^{j-1}))\\
        &= H(X^{i},Y^{j-1})+H(X^{i-1},Y^{j}) \\
        &\qquad- H(X^{i-1},Y^{j-1}) - H(X^{i},Y^{j}). 
    \end{align*}
    The formula for the Gaussian case \eqref{eq:cmi_entropies_gaussian} follows by using the definition of multivariate Gaussian entropy \cite{CovThom06}.

\subsection{Analysis of the Gap Between MI and Gaussian MI}
We provide an upper bound on the error of employing the Gaussian Mi term instead of $I(X;Y|Z)$.
We analyse the gap for arbitrary $(i,j)$ and for simplicity of presentation we denote $X_i=X$, $Y_j=Y$ and $(X^{i-1},Y^{j-1})=Z$.
Thus $\infomat_{i,j}=I(X;Y|Z)$.
We denote the with $(X_G,Y_G,Z_G)$ the jointly Gaussian triplet whose first and second moments are similar to those of $(X,Y,Z)$.
Consequently, the Gaussian conditional mutual information term corresponds to $I(X_G;Y_G|Z_G)$.
We utilize the following result from \cite{goldfeld2022k}
\begin{align*}
    &I(X;Y)-I(X_G;Y_G) \\
    &= \DKL(P_{X,Y}\|P_{X_G,Y_G}) - \DKL(P_{X}\otimes P_{Y}\|P_{X_G}\otimes P_{Y_G}).
\end{align*}
We have
\begin{align*}
    &I(X;Y|Z)-I(X_G;Y_G|Z_G)\\
    &\hspace{3cm}\leq \underbrace{I(X;Y|Z)-I(X_G;Y_G|Z)}_{:=\Delta_G} \\
    &\hspace{3cm}+ \underbrace{I(X_G;Y_G|Z) - I(X_G;Y_G|Z_G)}_{:=\Delta_Z}.
\end{align*}
We analyze each error term separately.
For $\Delta_G$ we have
\begin{align*}
    &\Delta_G = \int_{\cZ}\left(I(X;Y|Z=z) - I(X_G;Y_G|Z=z)\right)p_Z(z)\dd z\\
    &= \int_{\cZ}\Big( \DKL(P_{X,Y|Z=z}\|P_{X_G,Y_G|Z=z})\\
    &- \DKL(P_{X|Z=z}\otimes P_{Y|Z=z}\|P_{X_G|Z=z}\otimes P_{Y_G|Z=z}\Big)p_Z(z)\dd z\\
    &= \DKL(P_{X,Y|Z}\|P_{X_G,Y_G|Z}|P_Z)\\
    &-\DKL(P_{X|Z}\otimes P_{Y|Z}\|P_{X_G|Z}\otimes P_{Y_G|Z}|P_Z).
\end{align*}
To bound the second error, we insert an additional term, which considers the KL term conditioned on $Z_G$, while integrated w.r.t $P_{Z}$, given by
$$
\EE_Z\left[\DKL(P_{X_G,Y_G|Z_G}\|P_{Z_G|Z_G}\otimes P_{Y_G|Z_G}|Z_G)\right].
$$
We therefore have
\begin{align*}
    &\Delta_Z  = \int_{\cZ}\DKL(P_{X_G,Y_G|Z=z}\|P_{X_G|Z=z}\otimes P_{Y_G|Z=z})p_z(z)\dd z\\
    &-\int_{\cZ}\DKL(P_{X_G,Y_G|Z_G=z}\|P_{X_G|Z_G=z}\otimes P_{Y_G|Z_g=z})p_{z_G}(z)\dd z\\
    &= \int_{\cZ}\DKL(P_{X_G,Y_G|Z=z}\|P_{X_G|Z=z}\otimes P_{Y_G|Z=z})p_z(z) \\
    &-\int_{\cZ}\DKL(P_{X_G,Y_G|Z_G=z}\|P_{X_G|Z_G=z}\otimes P_{Y_G|Z_G=z})p_z(z)\\
    &+\int_{\cZ}\DKL(P_{X_G,Y_G|Z_G=z}\|P_{X_G|Z_G=z}\otimes P_{Y_G|Z_G=z})p_z(z)\\
    &-\int_{\cZ}\DKL(P_{X_G,Y_G|Z_G=z}\|P_{X_G|Z_G=z}\otimes P_{Y_G|Z_g=z})p_{z_G}(z)\dd z\\
    &= \DKL(P_{X_G,Y_G|Z}\|P_{X_G,Y_G|Z_G}|P_Z)\\
    &+\DKL(P_{X_G|Z}\otimes P_{Y_G|Z}\|P_{X_G|Z_G}\otimes P_{Y_G|Z_G}|P_Z)\\
    &=\int_{\cZ}\DKL(P_{X_G,Y_G|Z_G=z}\|P_{X_G|Z_G=z}\otimes P_{Y_G|Z_g=z})\\
    &\cdot(p_Z(z)-p_{z_G}(z))\dd z,
\end{align*}
where the last term can be upper bounded by
\begin{align*}
    &=\int_{\cZ}\DKL(P_{X_G,Y_G|Z_G=z}\|P_{X_G|Z_G=z}\otimes P_{Y_G|Z_g=z})\\
    &\hspace{3.5cm}\cdot(p_Z(z)-p_{z_G}(z))\dd z\\
    &\leq \max_{z\in\cZ}\DKL(P_{X_G,Y_G|Z_G=z}\|P_{X_G|Z_G=z}\otimes P_{Y_G|Z_g=z})\\
    &\hspace{3.5cm}\cdot d_{\mathsf{TV}}(p_Z,p_{z_G}).
\end{align*}
We note the the resulting upper bound extends the result from \cite{goldfeld2022k}. We also note that as desired $(\Delta_G+\Delta_Z)\to 0$ when 
$P_{X,Y,Z}\to P_{X_G,Y_G,Z_G}$.

% \subsection{Additional Information on MAF-based estimation}
% \begin{enumerate}
%     \item conditioner and transformer
%     \item algorithm
%     \item infomat algorithm
% \end{enumerate}

% \subsection{Additional Information on Plug-in Entropy Estimator}

\end{document}